\documentclass[11pt]{amsart}

\usepackage{mathrsfs}
\usepackage{amsfonts}
\usepackage{amssymb}
\usepackage{amsxtra}
\usepackage{dsfont}
\usepackage[dvipsnames]{xcolor}
\usepackage[english,polish]{babel}
\usepackage{enumerate}
\usepackage[shortlabels]{enumitem}
\usepackage{todonotes}

\usepackage{tikz}
\usetikzlibrary{arrows}

\usepackage[compress, sort]{cite}

\usepackage[T1]{fontenc}

\usepackage{newtxmath}
\usepackage[type1]{newtxtext}

\usepackage[margin=2.5cm, centering]{geometry}
\usepackage[colorlinks,citecolor=blue,urlcolor=blue,bookmarks=true]{hyperref}
\hypersetup{
pdfpagemode=UseNone,
pdfstartview=FitH,
pdfdisplaydoctitle=true,
pdfborder={0 0 0}, 
pdftitle={Spectral transitions in some Rabi models},
pdfauthor={Grzegorz Świderski and Lech Zieliński},
pdflang=en-US
}

\newcommand{\CC}{\mathbb{C}}
\newcommand{\ZZ}{\mathbb{Z}}

\newcommand{\NN}{\mathbb{N}}
\newcommand{\RR}{\mathbb{R}}

\newcommand{\calG}{\mathcal{G}}

\newcommand{\calD}{\mathcal{D}}
\newcommand{\calH}{\mathcal{H}}

\newcommand{\frakB}{\mathfrak{B}}
\newcommand{\frakX}{\mathfrak{X}}

\newcommand{\fraku}{\mathfrak{u}}
\newcommand{\frakt}{\mathfrak{t}}

\newcommand{\frakS}{\mathfrak{S}}

\newcommand{\sign}[1]{\operatorname{sign}({#1})}

\newcommand{\tr}{\operatorname{tr}}

\newcommand{\cl}[1]{\operatorname{cl}({#1})}

\newcommand{\sigmaEss}{\sigma_{\mathrm{ess}}}
\newcommand{\sigmaP}{\sigma_{\mathrm{p}}}
\newcommand{\sigmaSC}{\sigma_{\mathrm{sc}}}
\newcommand{\sigmaAC}{\sigma_{\mathrm{ac}}}

\newcommand{\sigmaD}{\sigma_{\mathrm{discr}}}

\newtheorem{theorem}{Theorem}[section]

\newtheorem{proposition}[theorem]{Proposition}

\theoremstyle{plain}
\newcounter{thm}

\numberwithin{equation}{section}

\theoremstyle{definition}

\newtheorem{Corollary}[theorem]{Corollary}

\title{Spectral transitions in some Rabi models}

\author{Grzegorz \'{S}widerski}
\address{ Grzegorz \'{S}widerski\\ Faculty of Pure and Applied Mathematics\\
 Wroc\l{}aw University of Science and Technology\\ Wyb. Wyspiańskiego 27\\50-370 Wroc\l{}aw, Poland}
\email{grzegorz.swiderski@pwr.edu.pl}

\author{Lech Zieli\'{n}ski}
\address{Lech Zieli\'{n}ski\\ Laboratoire de Mathématiques Pure et Appliquées\\
    Centre Universitaire de la Mi-Voix\\ 50 rue Ferdinand Buisson\\ CS 80699\\
    62228 CALAIS Cedex, France}
\email{Lech.Zielinski@univ-littoral.fr}


\begin{document} 
\selectlanguage{english} 

 \begin{abstract}  
We consider the transition from discrete to continuous spectrum which occurs in some quantum Rabi models. We present recent developments 
of the subordinacy theory and their applications 
to study the spectral properties of 
the intensity-dependent Rabi model, the anisotropic two-photon Rabi model, and the two-photon Rabi-Stark model for the whole range of parameters. 
Our analysis allows us to detect and locate the 
continuous spectrum. Moreover, we prove 
the absence of eigenvalues in the interior of the continuous spectrum and the absence of a singular spectrum for all models considered.

\end{abstract} 
\maketitle

\section{Introduction} \label{sec:0}
 The simplest model of interaction between light and matter was proposed in \cite{rabi1936process, rabi1937space} by I. I. Rabi. Its fully quantized version is called the  quantum Rabi model (QRM). The QRM couples a two-level system with quantized single-mode radiation and  was used to study real atoms in quantum optics, e.g. to experimentally test the field energy quantization of the cavity electrodynamics (see \cite{Bru}). 
   We refer to the survey \cite{xie2017quantum} for a list of research work and experimental realizations of the QRM and its generalizations that describe various  quantum devices. 
  \par 
   One of the generalizations is the two-photon QRM, 
   which describes the situation when the change of level is associated with the absorption/emission of two photons instead of one photon. 
 The corresponding two-photon QRM   was applied to describe a two-level atom interacting with squeezed light
(see  \cite{Ber, G}), quantum dots inserted in a cavity (see \cite{Va, Ota, Stufler}), trapped ions experiments (see \cite{Feli, Puebla}), and superconducting circuits (see \cite{Feli1, Feli2}). The two-photon QRM  
gave rise to many research works 
(see \cite{Braak:det, NG, Duan2016, Maciej, Penna, Lo!}). 
The great interest in this model was motivated by 
 the so-called spectral collapse phenomenon, i.e., the 
 transition from discrete to continuous spectrum that occurs when 
 the coupling 
 constant $g$ reaches a critical value $g_{\rm cr}$. 
 In particular, the spectrum of the two-photon QRM is discrete if $0<g<g_{\rm cr}$, but the spacing of the eigenvalues shrinks as $g$ approaches $g_{\rm cr}$ 
(see  \cite{Braak:det, Duan2016, Lo2021'}). 
A half-line of the continuous spectrum appears 
if $g=g_{\rm cr}$ (see \cite {Lo!}) and the spectrum 
becomes the whole real line if $g>g_{\rm cr}$ 
(see \cite{Braak:det}). 
This phenomenon was also studied numerically (see, e.g.  
 \cite{Polaron, Hammani, Lupo, NG, Armenta}) and  
  experimentally 
(see, e.g. \cite{Feli, Felicetti_2019}).  
Spectral transitions also occur 
for the anisotropic two-photon Rabi model 
(see \cite{Cui, Li-ani, Li_26, Lo-ani}), 
Rabi-Stark models  
(see \cite{ Johan, Braak-Stark, Chen-Stark, 
Li_25, 
 Li-Stark, Maciej-Stark, Xie-Chen, XieStark, Xie-Stark20, Yan}) 
and other generalizations of the two-photon Rabi model 
(see \cite{Duan2015, Duan_2022, Lo2020', Lo2021''',  Lo2021'', Lo2022, mixed, Ying, Ying_22, Ying_25}). We note that spectral transitions have recently 
 attracted  attention in connection with 
 the quantum metrology (see \cite{Ying_22, Ying_25}).  
\par  
In this paper, we consider Rabi models which can 
be expressed as a direct sum of operators 
defined by symmetric Jacobi matrices, 
i.e. symmetric 
tridiagonal matrices acting in the Hilbert space 
of square-summable sequences $\ell^2(\NN_0)$ 
(see Section \ref{sec:1''}). 
The spectral properties of the corresponding operators have been  the subject of numerous mathematical works. 
Our purpose is to present basic notions of the spectral theory and to apply recent developments of the subordinacy theory for periodically modulated Jacobi matrices. 
We note that this theory is based on the asymptotic 
analysis of the transfer matrix and was applied 
to investigate the spectral transitions in the fundamental 
work \cite{JANAS}. 
We have chosen to characterize the spectrum using 
the mathematical notions introduced in 
Section \ref{sec:1} and we have decided to avoid 
using the terms "bound states" and "the spectral collapse",   which require additional clarifications.  Our main results are presented in Section \ref{sec:2}. 
The  intensity-dependent Rabi model (see \cite{NG2000463, Rodriguez-Lara:14}) 
 is considered in Theorem \ref{thm:21}.  
 The two-photon Rabi model is considered 
 in Theorem \ref{thm:22} and the two-photon anisotropic Rabi model in Theorem \ref{thm:23}. Finally, 
in Theorem \ref{thm:24} we consider 
the two-photon Rabi-Stark model.\, 
The general scheme of our approach is described in 
Section \ref{sec:3}. In Section \ref{sec:4} we give 
the proofs of Theorems \ref{thm:21}-\ref{thm:24}. In Section \ref{sec:6} 
we describe the most important aspects of our research.

\section{Preliminaries} \label{sec:1} 
\subsection{Basic notions and notations} 
\label{sec:1'} 
Let $\calH$ be a complex Hilbert space equipped with 
the scalar product $\langle \cdot ,\cdot \rangle$ and let $H: \, \calD (H)\to \calH $ be a linear map defined on a dense subspace of $\calH$. 
The spectrum of $H$ is defined by the formula 
$\sigma (H)=\CC \setminus \rho (H)$, where $\rho (H)$ is the set of complex numbers $\lambda$ such that 
$(H-\lambda I)^{-1}$ exists and is a bounded operator on $\calH$. We say that $\lambda \in \CC$ 
is an eigenvalue of $H$ if the eigenspace 
$E_\lambda (H):=\{ x\in \calD (\calH): 
Hx=\lambda x\}$ is not equal to $\{0\}$ 
and  the multiplicity of $\lambda$ is 
defined as the dimension of $E_\lambda (H)$. 
The point spectrum of $H$ is  $$\sigma_{\rm p}(H)=
\{ \lambda \in \CC : \lambda \hbox{ is an eigenvalue of } H\}, $$
 the discrete spectrum of $H$ is
 $$\sigma_{\rm  discr}(H)=
\{ \lambda \in \sigma_{\rm p}(H) : 
\dim E_{\lambda}(H)<\infty \hbox{ and } 
\lambda \hbox{ is isolated from } 
\sigma (H)\setminus \{\lambda \} \} $$ 
 and the essential spectrum, 
$\sigmaEss (H)= \sigma (H)\setminus 
\sigma_{\rm discr}(H)$.  \par 
We say that $H$ closed if and only if its graph  
$\calG (H)=\{ (x,Hx)\in \calH \times \calH : 
x\in \calD (H) \}$ is closed in $\calH \times \calH$. 
If $\hat H: \, \calD (\hat H)\to \calH $ is a linear map defined on a dense subspace of $\calH$ and $\hat H$ is symmetric 
(i.e.  $\langle \hat Hx,y\rangle =\langle x,\hat H y\rangle$ for $x,y\in \calD (\hat H)$), then the closure of its graph in $\calH \times \calH$ is the graph of the symmetric operator called the closure of $\hat H$ (see, e.g. \cite{Schm}, Section 3.1). 
\par 
In what follows, $\NN_0$ is the set of non-negative 
integers, $\ell^2(\NN_0)$ is the complex 
Hilbert space of square-summable complex sequences 
$(x_n)_{n\in \NN_0 }$ 
equipped with the scalar product 
  \begin{equation} \label{not1a}
  \langle {x,y} \rangle_{\ell^2(\NN_0)}  =\sum_{n=0}^\infty \, 
  x_n {\overline {y_n}} 
  \end{equation} 
and $\ell^2_{\rm{fin}}(\NN_0)$ is the vector 
subspace of $\ell^2(\NN_0)$ composed of finite linear combinations of vectors from the canonical basis 
$\{ e_n\}_{ n\in \NN_0}$. We identify the Fock space 
with $\ell^2(\NN_0)$ and define the Hamiltonian of 
single-mode radiation, $\hat N$, 
as the closed linear operator in  $\ell^2(\NN_0)$ satisfying 
   \begin{equation}\label{not1d} 
  \hat N e_n  =n e_n \hbox{ for } n\in \NN_0 
\end{equation} 
We define the annihilation and  creation  operators, $\hat{a}$ and $\hat{a}^{\dagger}$, as linear operators 
$\calD(\hat N^{1/2})\to \ell^2(\NN_0)$ such that
   \begin{equation}\label{not1c}
  \hat{a}^{\dagger}\, e_n = \sqrt{n+1} \, e_{n+1} \hbox{\, for } n\in \NN_0 
  \end{equation} 
     \begin{equation}\label{not1c'} 
   \hat{a} \, e_0=0\hbox{ \, and \, }  
    \hat{a} \, e_n = \sqrt{n} \, e_{n-1} 
    \hbox{\, for }   n\in \NN_0 \setminus \{0\} 
   \end{equation} 
We note that $\hat N =\hat a^\dagger \hat a$.\, 
We also consider the linear operators  defined 
in $\CC^2$ by the matrices 
   $$
    I_2 :=\begin{pmatrix} 1 & 0 \\ 0 & 1 \end{pmatrix} ,\hspace{5mm} 
   \sigma_{x} :=\begin{pmatrix} 0 & 1 \\ 1 & 0 \end{pmatrix} ,\hspace{5mm} 
   \sigma_{z} :=\begin{pmatrix} 1 & 0 \\ 0 & -1 \end{pmatrix} ,\hspace{5mm} 
    \sigma_{+} :=\begin{pmatrix} 0 & 0 \\ 1 & 0 \end{pmatrix} ,\hspace{5mm} 
    \sigma_{-} :=\begin{pmatrix} 0 & 1 \\ 0 & 0 \end{pmatrix} 
   $$   
and we denote
\begin{equation} 
e^{1}:=\begin{pmatrix}
    1\\ 0
\end{pmatrix} \hspace{9mm} 
e^{-1}:=\begin{pmatrix}
    0\\ 1
\end{pmatrix}     
\end{equation} 
If $\nu =\pm 1$ and $g_+$, $g_-$ are real parameters, then we obtain  
\begin{equation} 
\sigma_x e^{\nu} = e^{-\nu}, \hspace{9mm} 
\sigma_z e^\nu =\nu e^\nu, \hspace{9mm} 
(g_- \sigma_- +g_+ \sigma_+ )e^{\pm 1} = g_\pm e^{\mp 1}\, .
\end{equation} 

In Section \ref{sec:2} we describe Rabi models defined by 
self-adjoint operators in the Hilbert space 
$\CC^2\otimes \ell^2(\NN_0)$ equipped with the 
canonical orthonormal basis 
$\{ e^\nu_n\}_{(\nu ,n)\in \{-1,1\} \times \NN_0}$  
of the form 
\begin{equation} \label{enun}
e^{\nu}_n :=e^\nu \otimes e_n .
\end{equation} 

\subsection{Jacobi operators} \label{sec:1''}
Assume that $(a_n)_{n\in \NN_0}$ and $(b_n)_{n\in \NN_0}$ 
are two sequences of positive and real numbers, respectively. Then  
the infinite tridiagonal symmetric matrix
\begin{equation}\label{Jacobimatrix} 
	\mathcal J ((a_n),(b_n)) = 
    \begin{pmatrix}
        b_0	& a_0	&	0	&       \\
        a_0 & b_1	& a_1	&         \\
         0   & a_1	& b_2 	& \ddots   \\
            &		&		\ddots &  \ddots 
    \end{pmatrix}
\end{equation} 
is called \emph{Jacobi matrix}. We let 
$\hat J((a_n),(b_n))$ denote the 
linear map  in 
$\ell^2_{\rm fin}(\NN_0)$,  defined by 
using complex valued sequences as column vectors, i.e. 
\begin{equation} \label{J}
\hat J((a_n),(b_n))e_n =a_ne_{n+1}+a_{n-1}e_{n-1}
+b_ne_n, 
\end{equation} 
where $a_je_j:=0$ when $j<0$. The map  
$\hat J((a_n),(b_n))$ is symmetric in 
the Hilbert space $\ell^2(\NN_0)$ 
and its closure defines the Jacobi operator 
$J((a_n),(b_n))$. \par 
Assume that the Jacobi operator $J=J((a_n),(b_n))$ is 
self-adjoint $\ell^2(\NN_0)$. It is well known   (see Section 2.5 in \cite{Teschl}) that 
$e_1$ is cyclic (i.e. span$\{ J^k e_1: k\in \NN_0 \}$ 
is dense in $\ell^2(\NN_0)$) and the operator 
$J$ is unitary similar to the operator of multiplication 
by $\lambda$ in $L^2(\RR ,d\rho (\lambda ))$ where 
$d\rho$ is a Borel measure on $\RR$. The measure 
$d\rho$ is determined by the conditions that 
its support is $\sigma (J)$  and  
 $$ 
\langle e_1, J^k e_1\rangle =
\int_{\RR} \lambda^k \, d\rho (\lambda ) 
\hbox{ for every } k\in \NN_0. 
$$ 
Consider the Lebesgue decomposition 
$d\rho =d\rho_{\rm {pp}}  +d\rho_{\rm{ac}} +
d\rho_{\rm{sc}}$, 
where pp, ac, and sc refer to the pure point, 
absolutely continuous, and singularly continuous 
part of the measure $d\rho$ with respect to Lebesgue 
measure. Then the pure point spectrum  
$\sigma_{\rm {pp}} (J)$, the absolutely continuous spectrum $\sigma_{\rm {ac}} (J)$ and the singular continuous spectrum  
$\sigma_{\rm {sc}} (J)$ are defined 
as the support of $d\rho_{\rm {pp}}$, 
$d\rho_{\rm{ac}}$ and $d\rho_{\rm{sc}}$ 
respectively. We note that $\sigma_{\rm {pp}}(J)$ is the closure of $\sigma_{\rm {p}}(J)$.
\par 
In what follows, we say that the self-adjoint operators 
$H$ and $H'$ are unitarily similar if and only if 
$H'=U^{-1}HU$ where $U$ is an isometric bijection 
between two Hilbert spaces. We will investigate Rabi models defined by a self-adjoint 
operator $H$ unitarily similar to a direct sum  $J_1 \oplus \dots \oplus J_l$, where 
$\{ J_k\}_{1\le k\le l}$ is a finite family of Jacobi operators. In this case 
$\sigma (H)=\sigma (J_1)\cup \dots \cup \sigma (J_l)$ and similar equalities  
remain valid if $\sigma$ is replaced by $\sigma_{\rm {pp}}$, $\sigma_{\rm {ac}}$ or 
$\sigma_{\rm {sc}}$ (see \cite{Schm}, Section 9.1). 

\section{Results} \label{sec:2}
\subsection{Intensity-dependent Rabi model} \label{sec:21} 
Assume $\kappa \ge 0$,\, $g>0$,  
$\Delta \in \RR$, and consider  the linear map in 
$ \CC^2 \otimes \ell^2_{\rm fin}(\NN_0)$ given by the formula
  \begin{equation}\label{21} 
 \hat H = I_2 \otimes \hat N 
 +\frac{\Delta}{2}\, \sigma_z \otimes I + 
  g \sigma_x \otimes \left( {(\hat N +2\kappa )^{1/2} \, \hat a +\hat a^\dagger \, (\hat N +2\kappa )^{1/2} }\right) 
  \end{equation} 
The operator $\hat H$ is symmetric in the Hilbert space 
$\CC^2 \otimes \ell^2(\NN_0)$ and  $H$, 
the Hamiltonian of the intensity-dependent QRM, is defined as   
the closure of $\hat H$ in $\CC^2 \otimes \ell^2(\NN_0)$. It is easy to see 
that the subspaces $\calH^+$ and 
$\calH^-$ spanned by $\{ e^1_0,e^{-1}_1,e^{1}_2,e^{-1}_3,
\dots \}$ and $\{ e^{-1}_0,e^{1}_1,e^{-1}_2,e^{1}_3,
\dots \}$ are invariant for  $H$ 
(see the action of $H$ on 
$e^{\pm 1}_n$ described in Section  
\ref{sec:41}) and   
$H$ can be written as a direct sum of two 
operators that can be expressed by 
Jacobi matrices. 
This decomposition is  
usually called {\it the parity decomposition}  and, it allows us to deduce the  
the properties of the spectrum of $H$ 
from the properties of corresponding 
Jacobi operators. 
\begin{theorem} \label{thm:21} 
The operator $H$ is unitarily similar to the direct sum \begin{equation} \label{decomp1}
J^-\oplus J^+ , 
\end{equation} 
where 
$J^\pm =J((a_n ),(b_n^\pm ))$ is the Jacobi operator with 
\begin{equation} \label{thm1} \begin{cases}
\hspace{2mm}  a_n = g \sqrt{(n+1)(n+ 2\kappa )},\\ 
\hspace{2mm}  b_n^{\pm} = n \pm (-1)^n \, \frac{\Delta}{2}
\end{cases} 
\end{equation} 
The operator $J^\pm$ is self-adjoint in 
$\ell^2(\NN_0)$ and 
\begin{equation} \label{thm1a}
    \sigma(J^\pm ) =\sigma_{\rm discr}(J^\pm ) 
    \hbox{ \, when } 
    0<g<\tfrac{1}{2}
\end{equation} 
\begin{equation} \label{thm1c}
\RR =\sigma(J^\pm )=  \sigmaAC(J^\pm ) \hbox{\, when } g>\tfrac{1}{2}
\end{equation} 
\begin{equation} \label{thm1b} 
[-\kappa , \infty )   =
    \sigmaAC(J^\pm )
= \sigma(J^\pm )\setminus 
\sigma_{\rm discr}(J^\pm )
  \hbox{\, when } g=\tfrac{1}{2}.
\end{equation} 
Moreover, in all cases, $\sigmaSC(J^\pm) = \emptyset$, and $\sigmaP(J^\pm)$ is disjoint from the interior of $\sigmaAC(J^\pm)$. 
\end{theorem}  
The statement of Theorem \ref{thm:21} ensures  
the following information on the spectral 
transition of the model: 
if $0<g<\frac 12$ then 
the whole spectrum of $H$ is 
discrete; if $g>\frac 12$ then the spectrum of $H$
is $\RR$ and $H$ has no eigenvalues; 
in the critical case $g=\frac 12$  the  spectrum of $H$ consists of 
the  half-line $[-\kappa , \infty )$ 
with no eigenvalues in $(-\kappa , \infty )$ and 
a possible discrete spectrum in $(-\infty ,-\kappa )$.

\subsection{Two-photon Rabi model} \label{sec:22} 
Assume $g>0$,  
$\Delta \in \RR$, and consider the linear map in 
$ \CC^2 \otimes \ell^2_{\rm fin}(\NN_0)$ given by
  \begin{equation}\label{22} 
  \hat H = I_2 \otimes \hat N +
  \frac{\Delta}{2}\, \sigma_z \otimes I +
  g \sigma_x \otimes \left( { \hat a^2 + 
  ({\hat a}^\dagger )^2 }\right) .  
  \end{equation} 
The operator $\hat H$ is symmetric in the Hilbert space 
$\CC^2 \otimes \ell^2(\NN_0)$ and  $H$, 
the Hamiltonian of the two-photon QRM, is defined as   
the closure of $\hat H$ in $\CC^2 \otimes \ell^2(\NN_0)$. Due to the definition of 
$\hat a^\dagger$ and $\hat a$ given in 
\eqref{not1c} and \eqref{not1c'}, we get 
 \begin{equation} \label{320}
  \big( \hat{a}^{\dagger}\big)^2\, e_n = \sqrt{(n+2)(n+1)} \, e_{n+2}, 
\hspace{8mm}  
\hat{a}^2 \, e_n = \sqrt{n(n-1)} \, e_{n-2} 
    \hbox{\, if }   n\ge 2,  
    \hspace{8mm}  
\hat{a}^2 \, e_n =0
    \hbox{\, if }   n\ge 1.  
 \end{equation} 
Due to \eqref{320}, the  
subspaces $\calH_0$ and $\calH_1$ spanned by 
$\{ e^\nu_{2n}\}_{(n,\nu )\in \NN_0 \times \{ -1,1\} }$ and 
$\{ e^\nu_{2n+1}\}_{(n,\nu )\in \NN_0 \times \{ -1,1\} }$ are invariant for $H$. 
Thus $H$ can be written as a direct sum 
of two operators usually called {\it the 
decomposition with respect to the 
Bargmann index} $q=\frac 14$ and 
$q=\frac 34$.
If $\mu =0$ 
or 1, then an additional parity decomposition $\calH_\mu =
\calH^+_\mu \oplus \calH^-_\mu$ with 
$\calH^+_\mu$ and $\calH^-_\mu$ 
spanned by 
$\{ e_{2n+\mu}^{(-1)^n} \}_{n\in \NN_0}$ and 
$\{ e_{2n+\mu}^{(-1)^n} \}_{n\in \NN_0}$,  
allows us to write $H$ as a direct sum of 
four  
operators that can be expressed by 
Jacobi matrices and the properties of the spectrum of $H$ 
can be deduced 
from the properties of the corresponding 
four Jacobi operators.

\begin{theorem} \label{thm:22} 
The operator $H$ is unitarily similar to the direct sum 
\begin{equation} \label{decomp2}
J_0^-\oplus J_0^+ \oplus J_1^- \oplus J_1^+  
\end{equation} where  
$J_\mu^\pm =J((a_{\mu ,n} ),(b_{\mu ,n}^\pm ))$ are the Jacobi operators with  
\begin{equation} \label{thm2} \begin{cases}
\hspace{2mm}  
a_{\mu ,n} = g\sqrt{(2n+1+\mu )(2n+2+\mu )},\\ 
\hspace{2mm}  b_{\mu ,n}^{\pm} = 2n+\mu \pm (-1)^n \, \frac{\Delta}{2}
 \end{cases}  
\end{equation} 
The operator $J^\pm_\mu$ is self-adjoint in $\ell^2(\NN_0)$ and 
\begin{equation} \label{thm2a} 
\sigma(J_\mu^\pm ) =\sigma_{\rm discr}(J_\mu^\pm )  
    \hbox{ \, when } 0<g<\tfrac{1}{2}
\end{equation} 
\begin{equation} \label{thm2b}
\RR =\sigma(J_\mu^\pm )= \sigmaAC(J_\mu^\pm ) \hbox{\, when } g>\tfrac{1}{2}
\end{equation} 
\begin{equation} \label{thm2c} 
[-\tfrac 12 , \infty ) =
    \sigmaAC(J_\mu^\pm )
= \sigma(J_\mu^\pm )\setminus 
\sigma_{\rm discr}(J_\mu^\pm )  
  \hbox{\, when } g=\tfrac{1}{2}
\end{equation} 
Moreover, in all cases, $\sigmaSC(J_{\mu}^\pm) = \emptyset$, and $\sigmaP(J_{\mu}^\pm)$ is disjoint from the interior of $\sigmaAC(J_{\mu}^\pm)$.  
\end{theorem}  
The statement of Theorem \ref{thm:22} ensures 
the following information on the spectral 
transition of the model: 
if $0<g<\frac 12$ then 
the whole spectrum of $H$ is 
discrete; if $g>\frac 12$ then the spectrum of $H$
is $\RR$ and $H$ has no eigenvalues; 
in the critical case $g=\frac 12$ the  spectrum of $H$ consists 
of the half-line $[-\tfrac 12 , \infty )$ with no eigenvalues in $(-\tfrac 12 , \infty )$ and a possible discrete spectrum in $(-\infty ,-\tfrac 12 )$.
\subsection{Anisotropic two-photon Rabi model} \label{sec:23} 
Assume $\Delta \in \RR$, $g_->0$, $g_+>0$, $g_- \ne g_+$, 
 and consider the linear map in 
$ \CC^2 \otimes \ell^2_{\rm fin}(\NN_0)$ given by 
the formula
  \begin{equation}\label{23} 
 \hat H = I_2 \otimes \hat N 
 +\frac{\Delta}{2}\, \sigma_z \otimes I + 
(g_-\sigma_- +g_+\sigma_+) \otimes (\hat a^\dagger )^2 +(g_+\sigma_- +g_-\sigma_+) \otimes \hat a^2
  \end{equation} 
Then $\hat H$ is symmetric in the Hilbert space 
$\CC^2 \otimes \ell^2(\NN_0)$ and  $H$, 
the Hamiltonian of the two-photon anisotropic Rabi model, 
is defined as   
the closure of $\hat H$ in $\CC^2 \otimes \ell^2(\NN_0)$. 
Similarly as in Section \ref{sec:22}, the operator $H$ can be written as a direct sum of four operators that can be expressed by Jacobi matrices and the properties of the
spectrum of $H$ can be deduced from the properties of the corresponding four Jacobi operators.
\begin{theorem} \label{thm:23} 
The operator $H$ is unitarily similar to the direct sum 
 \eqref{decomp2}, where 
$J_\mu^\pm =J((a_{\mu ,n}^\pm ),(b_{\mu ,n}^\pm ))$ are   Jacobi operators with  
\begin{equation} \label{thm3} 
\begin{cases}
\hspace{2mm}  
a_{\mu ,n}^\pm = (g\mp (-1)^n\, g') 
\sqrt{(2n+1+\mu )(2n+2+\mu )},\\ 
\hspace{2mm}  b_{\mu ,n}^{\pm} = 2n+\mu \pm (-1)^n \, \frac{\Delta}{2}
 \end{cases}  
\end{equation} 
where 
\begin{equation} \label{thm3g} 
g:=\frac{g_+ +g_-}{2} , \quad g':=\frac{g_+ -g_-}{2} .
\end{equation} 
The operators $J^\pm_\mu$ are self-adjoint in $\ell^2(\NN_0)$ and 
\begin{equation} \label{thm3c}
   \sigma(J_\mu^\pm ) =\sigma_{\rm discr}(J_\mu^\pm )  
    \hbox{ \, when }  g< \tfrac{1}{2}
\end{equation} 
\begin{equation} \label{thm3d}
  \sigma(J_\mu^\pm ) =\sigma_{\rm discr}(J_\mu^\pm )   \hbox{\, when } |g'|> \tfrac{1}{2}
\end{equation} 
\begin{equation} \label{thm3a}
   \RR =\sigma(J_\mu^\pm )= \sigmaAC(J_\mu^\pm ) 
    \hbox{ \, when } |g'|<\tfrac{1}{2}<g
\end{equation} 
\begin{equation} \label{thm3b} 
    [-\tfrac{1}{2}, \infty ) 
    =
\sigmaAC (J_\mu^\pm ) 
    =  \sigma (J_\mu^\pm ) 
\setminus \sigmaD(J_\mu^\pm )
      \hbox{\, when } g= \tfrac{1}{2}
\end{equation} 
\begin{equation} \label{thm3e}
(-\infty, -\tfrac{1}{2}] 
=
    \sigmaAC(J_\mu^\pm ) 
=  \sigma (J_\mu^\pm ) 
\setminus \sigmaD(J_\mu^\pm )
     \hbox{\, when } |g'|= \tfrac{1}{2}
\end{equation} 
Moreover, in all cases, $\sigmaSC(J_{\mu}^\pm) = \emptyset$, and $\sigmaP(J_{\mu}^\pm)$ is disjoint from the interior of $\sigmaAC(J_{\mu}^\pm)$.
\end{theorem}  
We discuss the above statement, 
following the notation 
of \cite{Li-ani}. Assume  
$0<g_+<g_-$ and denote $g:=g_-$,\, 
$r:={g_+}/{g_-} \in (0,1)$. Then 
Theorem \ref{thm:23} ensures the following:
\begin{enumerate}
    \item If $g<\frac{1}{1+r}$, 
then the whole spectrum of $H$ is discrete
\item If $g=:g_{\rm{cr}}=\frac{1}{1+r}$, then the spectrum of $H$ consists of the half-line $[-\frac 12,\infty )$ 
with no eigenvalues in $(-\frac 12,\infty )$ and a possible discrete spectrum in 
$(-\infty ,-\frac 12)$
\item If $\frac{1}{1+r} < g 
< \frac{1}{1-r}$, 
then the spectrum of $H$ is $\RR$ and there is no eigenvalue  
\item If $g=g'_{\rm{cr}}=\frac{1}{1-r}$, 
then  the spectrum of 
$H$ consists of the half-line $(-\infty ,-\frac 12]$ 
with no eigenvalues in $(-\infty ,-\frac 12 )$ and a possible discrete spectrum in $(-\frac 12,\infty )$ 
\item If $g>\frac{1}{1-r}$, 
then  the whole spectrum of 
$H$ is discrete  
\end{enumerate}
We note that the cases (4), (5)    
have not been investigated 
in \cite{Cui, Li-ani, Li_26, Lo-ani}.  
Only the case (1) was investigated numerically. 

\subsection{Two-photon Rabi-Stark quantum model} \label{sec:24} 
Assume $g>0$,  $\Delta \in \RR$, $\kappa \in \RR$ and consider the linear map in 
$\CC^2 \otimes \ell^2_{\rm fin}(\NN_0)$ given by the formula
  \begin{equation}\label{24} 
  \hat H = I_2 \otimes \hat N 
  +\, \sigma_z \otimes 
  \Big( {\kappa \hat N +\frac{\Delta}{2}}\Big)  + 
  g \sigma_x \otimes \left( { \hat a^2 + 
  ({\hat a}^\dagger )^2 }\right) 
  \end{equation} 
The operator $\hat H$ is symmetric in the Hilbert space 
$\CC^2 \otimes \ell^2(\NN_0)$ and  $H$, 
the Hamiltonian of the two-photon Rabi-Stark quantum model, 
is defined as   
the closure of $\hat H$ in $\CC^2 \otimes \ell^2(\NN_0)$. Similarly as in Sections  \ref{sec:22} and \ref{sec:23}, the operator $H$ can be written as a direct sum of four operators that can be expressed by Jacobi matrices and the properties of the
spectrum of $H$ can be deduced from the properties of the corresponding four Jacobi operators. 
\begin{theorem} \label{thm:24} 
The operator $H$ is unitarily similar to the direct sum 
 \eqref{decomp2}, where 
$J_\mu^\pm =J((a_{\mu ,n} ),(b_{\mu ,n}^\pm ))$ are   Jacobi operators with  
\begin{equation} \label{thm4} 
\begin{cases}
\hspace{2mm}  
a_{\mu ,n} = g
\sqrt{(2n+1+\mu )(2n+2+\mu )},\\ 
\hspace{2mm}  b_{\mu ,n}^{\pm} = (2n+\mu )
\left({ 1 \pm (-1)^n \kappa }\right)
\pm (-1)^n \frac{\Delta}{2}
 \end{cases}  
\end{equation} 
The operators $J^\pm_\mu$ are self-adjoint in $\ell^2(\NN_0)$ and 
\begin{equation} \label{thm4a} \sigma(J_\mu^\pm ) =\sigma_{\rm discr}(J_\mu^\pm )  
    \hbox{ \, when } |\kappa |>1
\end{equation} 
\begin{equation} \label{thm4e}
   \sigma(J_\mu^\pm ) =\sigma_{\rm discr}(J_\mu^\pm )  
    \hbox{ \, when } \kappa^2 +4g^2<1
\end{equation} 
 \begin{equation} \label{thm4c}
     \RR =\sigma(J_\mu^\pm )= \sigmaAC(J_\mu^\pm ) 
    \hbox{ \, when } |\kappa |<1 \hbox{ and } 
    \kappa^2 +4g^2>1
\end{equation} 
\begin{equation} \label{thm4d}
    [ \tfrac {\kappa^2 -1 - \kappa \Delta}{2}
    , \infty ) =
\sigmaAC (J_\mu^\pm ) 
    =  \sigma (J_\mu^\pm ) 
\setminus \sigmaD(J_\mu^\pm )
    \hbox{ \, when } \kappa^2 +4g^2=1
\end{equation} 
\begin{equation} \label{thm4b} 
  (-\infty, -\tfrac { \kappa \Delta}{2} ] 
  =
\sigmaAC (J_\mu^\pm ) 
   =  \sigma (J_\mu^\pm ) 
\setminus \sigmaD(J_\mu^\pm ) 
    \hbox{ \, when }  |\kappa |=1
 \end{equation}  
Moreover, in all cases, $\sigmaSC(J_{\mu}^\pm) = \emptyset$, and $\sigmaP(J_{\mu}^\pm)$ is disjoint from the interior of $\sigmaAC(J_{\mu}^\pm)$.
\end{theorem}  
The statement of Theorem \ref{thm:24} ensures the following: 
\begin{enumerate}
    \item the whole spectrum of $H$ is discrete  
in the case $\kappa^2+4g^2<1$ and 
in the case $|\kappa |>1$
\item if $|\kappa |<1$ and 
$\kappa^2+4g^2>1$ then the spectrum of $H$ is $\RR$ and there is no eigenvalue 
\item in the critical case 
$\kappa^2+4g^2=1$ the spectrum of 
$H$ consists of the half-line 
$[ \tfrac {\kappa^2 -1 - \kappa \Delta}{2}, \infty )$ 
with no eigenvalues in 
$(\tfrac {\kappa^2 -1 - \kappa \Delta}{2}, \infty )$ and a possible discrete spectrum in 
$(-\infty ,\tfrac {\kappa^2 -1 - \kappa \Delta}{2})$
\item in the critical case $|\kappa |=1$ the spectrum of 
$H$ consists of the half-line 
$(-\infty ,-\tfrac { \kappa \Delta}{2}]$ 
with no eigenvalues in $(-\infty ,-\tfrac { \kappa \Delta}{2})$ and a possible discrete spectrum in 
$(-\tfrac { \kappa \Delta}{2},\infty )$ 
\end{enumerate}
We note that this model was investigated 
theoretically and numerically by J. Li, Q.-H. Chen 
\cite{Li-Stark, Li_25}, assuming $|\kappa |<1$. It seems that other cases  have not 
been investigated by now. 

\section{The general scheme} \label{sec:3} 
\subsection{A criterion for self-adjointness} 
\label{sec:31} 
We begin by the following well-known result 
(see, e.g. \cite[Corollary 6.19]{Schmudgen2017}) 
\begin{theorem}[Carleman] \label{Carleman}
Let $(a_n)$ and $(b_n)$ 
be two sequences of real numbers. If $J((a_n),(b_n))$ 
is the closure of the linear operator defined in 
$\ell^2_{\rm fin}(\NN_0)$ by \eqref{J} 
and 
\begin{equation} \label{eq:Carleman}
    \sum_{n=0}^\infty \frac{1}{|a_n|} =\infty 
 \end{equation} 
then $J((a_n),(b_n))$ is self-adjoint in $\ell^2(\NN_0)$. 
\end{theorem} 
The assertion of Theorem \ref{Carleman} ensures the 
self-adjointness of the operators $H$ introduced in Sections \ref{sec:21}-\ref{sec:24}. This fact results  
from the decompositions \eqref{decomp1}, \eqref{decomp2}, and the following
\begin{Corollary}
The operators $J^\pm$ (respectively $J^\pm_\mu$) 
introduced in  
Theorem \ref{thm:21} (respectively Theorem \ref{thm:22}, 
\ref{thm:23} or Theorem \ref{thm:24}) are self-adjoint. 
\end{Corollary}
\begin{proof} Each operator in question is defined as 
the closure of the operator defined in 
$\ell^2_{\rm fin}(\NN_0)$ by  
\eqref{J} and 
$$
|a_n |\le c(n+1) 
$$
holds with a certain $c>0$, hence  $(a_n)$ 
satisfies the Carleman's condition \eqref{eq:Carleman}. 
\end{proof}

\subsection{Stolz class}
Let $N$ be a positive integer. We say that a sequence $(x_n : n \geq 1)$ belongs to $\calD_1^N$ if
\[
    \sum_{n=1}^\infty |x_{n+N} - x_n| < \infty.
\]
Notice that $(x_n : n \geq 1) \in \calD_1^N$ if and only if for any $i \in \{0,1,\ldots,N-1\}$ the sequence $(x_{nN+i} : n \geq 1)$ belongs to $\calD_1^1$, so, in particular, it is convergent. Moreover, we have the following
\begin{proposition} \label{prop:3}
Suppose that $(x_n),(y_n) \in \calD_1^N$. Then
\begin{itemize}
\item for any $\alpha \in \CC$ we have $(\alpha \cdot x_n) \in \calD_1^N$,
\item $(x_n + y_n) \in \calD_1^N$,
\item $(x_n \cdot y_n) \in \calD_1^N$,
\item suppose that $f$ is a Lipschitz continuous function on a compact interval $[a,b]$. If for some $M$ we have $\{x_n : n \geq M\} \subset [a,b]$, then $(f(x_n)) \in \calD_1^N$.
\end{itemize}
\end{proposition}
The proof of Proposition~\ref{prop:3} is straightforward.

\subsection{Periodic modulations}
Let $N$ be a positive integer. We say that Jacobi parameters $(a_n),(b_n)$ are $N$-periodically modulated if there exist $N$-periodic sequences $(\alpha_n : n \in \ZZ),(\beta_n : n \in \ZZ)$ of positive and real numbers, respectively, such that
\begin{equation}
	\label{eq:5}
    \lim_{n \to \infty} \bigg| \frac{a_{n-1}}{a_n} - \frac{\alpha_{n-1}}{\alpha_n} \bigg| = 0, \quad 
    \lim_{n \to \infty} \bigg| \frac{b_n}{a_n} - \frac{\beta_n}{\alpha_n} \bigg| = 0, \quad
    \lim_{n \to \infty} a_n = \infty.
\end{equation}
\begin{proposition} \label{prop:2}
If Jacobi parameters satisfy 
\begin{equation} 
    \label{eq:1}
    \bigg( \frac{\alpha_{n-1}}{\alpha_n} a_n - a_{n-1} \bigg),
    \bigg( \frac{\beta_{n}}{\alpha_n} a_n - b_n \bigg),
    \bigg( \frac{1}{\sqrt{a_n}} \bigg) \in \calD_1^N,
\end{equation}
then
\begin{equation}
    \label{eq:1'}
    \bigg( \frac{a_{n-1}}{a_n} \bigg),
    \bigg( \frac{b_n}{a_n} \bigg),
    \bigg( \frac{1}{a_n} \bigg) \in \calD_1^N.
\end{equation}
\end{proposition}
\begin{proof}
We shall repeatedly use Proposition~\ref{prop:3}. Suppose that \eqref{eq:1} is satisfied. Thus
\[
    \frac{1}{a_n} = \frac{1}{\sqrt{a}_n} \frac{1}{\sqrt{a}_n}
\]
also belongs to $\calD_1^N$. Next, since
\begin{align*}
    \frac{a_{n-1}}{a_n} 
    &= 
    -\frac{1}{a_n} \Big( \frac{\alpha_{n-1}}{\alpha_n} a_n - a_{n-1} \Big) + \frac{\alpha_{n-1}}{\alpha_n} \\
    \frac{b_n}{a_n} 
    &=
    -\frac{1}{a_n} \Big( \frac{\beta_{n}}{\alpha_n} a_n - b_{n} \Big) + \frac{\beta_{n}}{\alpha_n} 
\end{align*}
these sequences also belong to $\calD_1^N$, which ends the proof.
\end{proof}
\subsection{Main tools}
For any $n \in \ZZ$ let us define
\[
    \frakX_n(x) =
    \frakB_{n+N-1}(x) \ldots \frakB_{n+1}(x) \frakB_n(x), 
    \quad \text{where} \quad
    \frakB_n(x) = 
    \begin{pmatrix}
        0 & 1 \\
        -\frac{\alpha_{n-1}}{\alpha_n} & \frac{x-\beta_n}{\alpha_n}
    \end{pmatrix}.
\]
Spectral properties of Jacobi matrices with $N$-periodically modulated parameters depend on $\tr \frakX_0(0)$. 
The following theorem follows from \cite[Corollary 8]{SwiderskiTrojan2019} (see also \cite[Theorem 1]{PeriodicII}).
\begin{theorem} \label{thm:perI}
Suppose that Jacobi parameters $(a_n),(b_n)$ are $N$-periodically modulated and $\tr \frakX_0(0) \in (-2,2)$. Assume further that
\[
    \bigg( \frac{a_{n-1}}{a_n} \bigg),
    \bigg( \frac{b_n}{a_n} \bigg),
    \bigg( \frac{1}{a_n} \bigg) \in \calD_1^N.
\]
Then $J$ is self-adjoint if and only if 
the Carleman's condition 
\eqref{eq:Carleman} is satisfied. If that is the case, then 
\[
    \sigmaAC(J) = \RR, \quad{and} \quad \sigmaSC(J) = \emptyset, \quad{and} \quad \sigma_{\rm {p}}(J) = \emptyset.
\]
\end{theorem}

The following theorem is a consequence of \cite[Theorem A]{jordan2}, but its hypotheses are taken from a less general \cite[Theorem A]{jordan}.
\begin{theorem} \label{thm:perIIb}
Suppose that Jacobi parameters $(a_n),(b_n)$ are $N$-periodically modulated and $\frakX_0(0)$ is not diagonalizable\footnote{In particular, $\tr \frakX_0(0) \in \{-2,2\}$ and $\frakX_0(0)$ is not a multiple of the identity matrix.}. Assume further that
\begin{equation} 
	\label{eq:4}
    \bigg( \frac{\alpha_{n-1}}{\alpha_n} a_n - a_{n-1} \bigg),
    \bigg( \frac{\beta_{n}}{\alpha_n} a_n - b_n \bigg),
    \bigg( \frac{1}{\sqrt{a_n}} \bigg) \in \calD_1^N.
\end{equation}
Then $J$ is self-adjoint. Define $N$-periodic sequences $(s_n),(r_n)$ by
\[
    \lim_{n \to \infty} 
    \bigg| \frac{\alpha_{n-1}}{\alpha_n} a_n - a_{n-1} - s_n \bigg| = 0, 
    \quad \text{and} \quad
    \lim_{n \to \infty} 
    \bigg| \frac{\beta_{n}}{\alpha_n} a_n - b_n - r_n \bigg| = 0.
\]
Let $\varepsilon = \sign{\tr \frakX_0(0)}$ and define a polynomial
\begin{equation}
    \label{eq:3}
    \tau(x) = 
    \sum_{i=0}^{N-1} 
    \bigg(
    \frac{s_i}{\alpha_{i-1}} \big( 1-\varepsilon [\frakX_i(0)]_{1,1} \big) - 
    \frac{x+r_i}{\alpha_{i-1}} \varepsilon [\frakX_i(0)]_{2,1} 
    \bigg).
\end{equation}
Then
\[
    \sigmaAC(J) = \sigma(J) \setminus \sigmaD(J) = \cl{\tau^{-1} \big( (-\infty,0) \big)}, 
    \quad \text{and} \quad 
    \sigmaSC(J) = \emptyset,
    \quad \text{and} \quad
    \sigmaP(J) \cap \tau^{-1}\big( (-\infty,0) \big) = \emptyset.
\]
\end{theorem}
\begin{proof}
We are going to show that the hypotheses of \cite[Theorem A]{jordan2} are satisfied for $\gamma_n=a_n$. Notice that
\begin{align}
	\label{eq:9a}
	\sqrt{a_n} \bigg( \frac{\alpha_{n-1}}{\alpha_n} - \frac{a_{n-1}}{a_n} \bigg)
	&=
	\frac{1}{\sqrt{a_n}} \bigg( \frac{\alpha_{n-1}}{\alpha_n} a_n - a_{n-1} \bigg) \\
	\label{eq:9b}
	\sqrt{a_n} \bigg( \frac{\beta_{n}}{\alpha_n} - \frac{b_{n}}{a_n} \bigg)
	&=
	\frac{1}{\sqrt{a_n}} \bigg( \frac{\beta_{n}}{\alpha_n} a_n - b_{n} \bigg).
\end{align}
Thus, in view of \eqref{eq:4} the sequences on the left-hand sides belong to $\calD_1^N$. Next, by Proposition~\ref{prop:2} we have that \eqref{eq:1'} is satisfied. In view of \eqref{eq:5} and Proposition~\ref{prop:3} we have
\[
	\bigg( \sqrt{\frac{a_{n-1}}{a_n}} \bigg) \in \calD_1^N.
\]
Since
\[
	\sqrt{a_n} \bigg( \sqrt{\frac{\alpha_{n-1}}{\alpha_n}} - \sqrt{\frac{a_{n-1}}{a_n}} \bigg)
	=
	\sqrt{a_n} \bigg( \frac{\alpha_{n-1}}{\alpha_n} - \frac{a_{n-1}}{a_n} \bigg) 
	\bigg( \sqrt{\frac{\alpha_{n-1}}{\alpha_n}} + \sqrt{\frac{a_{n-1}}{a_n}} \bigg)^{-1}
\]
the sequence on the left-hand side belongs to $\calD_1^N$. Consequently, we have shown \cite[formula (1.2)]{jordan2}. According to \cite[formula (2.9)]{jordan} the sequence $(a_{n+N} - a_n)$ is bounded. Therefore,
\[
	\lim_{n \to \infty} (\sqrt{a_{n+N}} - \sqrt{a_{n}}) = 
	\lim_{n \to \infty} \frac{a_{n+N} - a_n}{\sqrt{a_{n+N}} + \sqrt{a_n}} = 0
\]
and we have verified that \cite[formula (1.3)]{jordan2} holds true. Finally, the condition~\eqref{eq:4} easily implies \cite[formula (1.4)]{jordan2}. Therefore, we have shown that the hypotheses of \cite[Theorem A]{jordan2} are satisfied.

It remains to compare the formula \eqref{eq:3} with \cite[formula (2.17)]{jordan2}. By \eqref{eq:9a} we have
\[
	\lim_{n \to \infty} 
	\sqrt{a_n} \bigg( \frac{\alpha_{n-1}}{\alpha_n} - \frac{a_{n-1}}{a_n} \bigg) = 0,
\]
which leads to $\frakS = 0$ (cf. \cite[formulas (2.19) and (2.6)]{jordan2}). Next, $\frakt=1$ (cf. \cite[formula (2.8)]{jordan2}). Finally, we have
\[
	\fraku_n = 
	s_n \big( 1-\varepsilon [\frakX_n(0)]_{1,1} \big) - 
    r_n \varepsilon [\frakX_n(0)]_{2,1} 
\]
(cf. \cite[formula (2.16)]{jordan2}), which leads to the equality of \eqref{eq:3} and \cite[formula (2.17)]{jordan2}.
\end{proof}

The following theorem follows from \cite[Theorem 5.3 and Remark 5.4]{Discrete}.
\begin{theorem} \label{thm:perIII}
Suppose that Jacobi parameters $(a_n),(b_n)$ are $N$-periodically modulated and $\tr \frakX_0(0) \in \RR \setminus [-2,2]$. Assume further that
\[
    \bigg( \frac{a_{n-1}}{a_n} \bigg),
    \bigg( \frac{b_n}{a_n} \bigg),
    \bigg( \frac{1}{a_n} \bigg) \in \calD_1^N.
\]
Then $J$ is self-adjoint and  
$\sigmaEss(J) =\sigma (J)\setminus 
\sigmaD(J)=\emptyset$.
\end{theorem}

\subsection{An auxiliary result}
The following proposition will be instrumental for our studies of Rabi models. 
\begin{proposition} \label{prop:1}
Let $N$ be a positive integer. Consider
\[
    a_n = \alpha_n \sqrt{(n+t)(n+s)}, \quad
    b_n = \beta_n n + \gamma_n,
\]
where $t,s>0$, $(\alpha_n),(\beta_n),(\gamma_n)$ are $N$-periodic sequences with $(\alpha_n)$ positive and $(\beta_n),(\gamma_n)$ real. Then \eqref{eq:1} holds true. Moreover, 
\begin{equation}
    \label{eq:2}
    \lim_{n \to \infty} 
    \bigg| \frac{\alpha_{n-1}}{\alpha_n} a_n - a_{n-1} - s_n \bigg| = 0, 
    \quad \text{and} \quad
    \lim_{n \to \infty} 
    \bigg| \frac{\beta_{n}}{\alpha_n} a_n - b_n - r_n \bigg| = 0
\end{equation}
holds for
\[
    s_n = \alpha_{n-1}, \quad r_n = \frac{\beta_n}{2} (t+s) - \gamma_n.
\]
\end{proposition}
\begin{proof}
Consider a sequence
\[
    \tilde{a}_n = \sqrt{(n+t)(n+s)}, \quad n \geq 0.
\]
It is immediate that
\begin{equation} \label{eq:6}
    \lim_{n \to \infty} (\tilde{a}_n - \tilde{a}_{n-1}) = 1.
\end{equation}
Let us define
\[
    f(x) = \sqrt{(x+t)(x+s)}, \quad x > \max(-t,-s).
\]
Notice that 
\begin{equation} \label{eq:7}
    f'(x) = \frac{2x+s+t}{2\sqrt{(x+t)(x+s)}}, \quad
    f''(x) = - \frac{(t-s)^2}{((x+t)(x+s))^{3/2}}.
\end{equation}
Thus the function $f$ is increasing and concave (the case $t=s$ corresponds to a linear function). Since $\tilde{a}_n = f(n)$ by \cite[Proposition 10.2]{jordan} we have
\[
    (\tilde{a}_n - \tilde{a}_{n-1}), \Big( \frac{1}{\sqrt{\tilde{a}_n}} \Big) \in \calD_1^1.
\]
Thus, these sequences belong also to $\calD_1^N$. Notice 
\[
    a_n = \alpha_n \tilde{a}_n, \quad
    \frac{\alpha_{n-1}}{\alpha_n} a_n - a_{n-1} = \alpha_{n-1}(\tilde{a}_n - \tilde{a}_{n-1}).
\]
Since $\calD_1^N$ is an algebra we get
\[
    \Big( \frac{\alpha_{n-1}}{\alpha_n} a_n - a_{n-1} \Big), 
    \Big( \frac{1}{\sqrt{a_n}} \Big) \in \calD_1^N.
\]
Moreover, in view of \eqref{eq:6} we get the first equality in \eqref{eq:2}. Next, define $g(x) = f(x) - x$. In view of \eqref{eq:7} we have
\[
    g'(x) = \frac{2x+t+s}{2\sqrt{(x+t)(x+s)}} -1.
\]
By direct computations we get that $g'$ is strictly positive if and only if $t \neq s$ (if $t=s$, then $g$ is a constant function). Therefore, $g$ is non-decreasing. Moreover,
\begin{equation} \label{eq:8}
    \lim_{x \to \infty} g(x) = \frac{t+s}{2}.
\end{equation}
Notice $g(n) = \tilde{a}_n - n$. Therefore, $(\tilde{a}_n - n)$ is non-decreasing and bounded, thus it belongs to $\calD_1^1$, so also to $\calD_1^N$. Since
\[
    \frac{\beta_n}{\alpha_n} a_n - b_n = \beta_n (\tilde{a}_n - n) - \gamma_n
\]
we have that this sequence also belongs to $\calD_1^N$. In view of \eqref{eq:8} the second equality in \eqref{eq:2} follows. The proof of \eqref{eq:1} is complete.
\end{proof}

\section{Proofs of Theorems \ref{thm:21}-\ref{thm:24}} \label{sec:4}
\subsection{Proof of Theorem \ref{thm:21}} \label{sec:41}  In this section, $H$ is the Hamiltonian of the 
intensity-dependent Rabi model defined as 
the closure of $\hat H$ given in \eqref{21}. 
\par  
Let $\{ e^\nu_n \}_{(\nu ,n)\in \{ -1,1\} \times \NN_0}$ be 
 the basis of $\CC^2\otimes \ell^2(\NN_0)$ given by \eqref{enun}. Since $e^\nu_n=e^\nu \otimes e_n$ and 
 $\sigma_x e^\nu =e^{-\nu}$, 
$\sigma_z e^\nu =\nu e^\nu$, we have
$$
g\big(\sigma_x \otimes \hat a^\dagger 
(\hat N+2\kappa )^{1/2} 
\big) e_n^{\nu} =g\big(\sigma_x e^\nu \otimes 
\hat a^\dagger (\hat N+2\kappa )^{1/2} e_n 
\big)= a_n\, e^{-\nu}_{n+1}, 
$$ 
$$
g\big(\sigma_x \otimes (\hat N+2\kappa )^{1/2} \hat a 
 \big) e_n^{\nu} =g\big(\sigma_x e^\nu \otimes 
(\hat N+2\kappa )^{1/2} \hat a e_n 
\big)= a_{n-1}\, e^{-\nu}_{n-1}
$$ 
with $(a_n)$ given by \eqref{thm1}. Moreover, 
$$
 \Big( I_2 \otimes \hat N +\frac{\Delta}{2}\, \sigma_z \otimes I\Big) e_n^{\nu} =\left( {n+\nu \frac{\Delta}{2} 
 }\right) e_n^{\nu}
 $$
and, introducing 
$\mathfrak{F}^\pm :=\{ f^\pm_n\}_{n\in \NN_0}$
with 
\begin{equation} \label{41B} 
 f^\pm_n:=e_n^{\pm (-1)^n}, 
\end{equation} 
we obtain
\begin{equation} \label{41H}
\hat H f^\pm_n =a_nf^\pm_{n+1} +a_{n-1}f^\pm_{n-1} 
+ b^{\pm}_n f^\pm_n,   
\end{equation} 
where $(b_n^\pm )$ is given by \eqref{thm1} and  $a_{j}f^\pm_{j}:=0$ when $j<0$. In order to obtain \eqref{decomp1}, we consider the orthogonal decomposition 
\begin{equation} \label{12orth}
\CC^2\otimes \ell^2(\NN_0) =\mathcal H^- 
\oplus \mathcal H^+
\end{equation} 
with $\mathcal H^\pm$ defined as the closure of 
${\rm span}(\mathfrak F^\pm )$. 
Due to \eqref{41H}, $\mathcal H^\pm$ are 
invariant subspaces for $\hat H$ and 
 $J((a_{n}),(b_{n}^\pm ))$ 
is the matrix of $\hat H|_{\mathcal H^\pm}$
in the basis $\mathfrak F^\pm$. 
 \par 
To begin the analysis of $J^\pm$ defined by 
$J((a_n),(b_n^\pm ))$, we write  
$$
a_n = \alpha_n 
    \sqrt{(n+1) (n+2\kappa )} \hbox{ \, with \, } 
    \alpha_n \equiv g, 
$$ 
$$
b_n^\pm =\beta_n n +\gamma_n^\pm  
\hbox{ \, with \, } \beta_n \equiv 1, \hspace{2mm} \gamma_n^\pm = \pm (-1)^n \frac{\Delta}{2} 
$$ 
Then \eqref{eq:1} and \eqref{eq:1'} are satisfied with $N=2$ due to Proposition~\ref{prop:1} and \ref{prop:2}.  
Moreover, 
\begin{equation}
    \label{eq:100}
    {\mathfrak B}_n(0) =
    \begin{pmatrix}
        0 & 1 \\
        -1 & -\frac{1}{g} 
    \end{pmatrix} , \quad 
    \frakX_n(0) = 
    {\mathfrak B}_n(0)^2 =
    \begin{pmatrix}
        -1 & -\frac{1}{g} \\
        \frac{1}{g} & -1 + \frac{1}{g^2}
    \end{pmatrix}  
\end{equation}
and 
 \begin{equation} \label{1tr}
    \tr \frakX_0(0) = -2 + \frac{1}{g^2}.
\end{equation}
\subsubsection{Case $0<g<\frac 12$} In this case, \eqref{1tr}
gives $\tr \frakX_0(0) > 2$,  and 
Theorem~\ref{thm:perIII} ensures 
$\sigma(J^\pm )=\sigmaD(J^\pm )$. 
\subsubsection{Case $g>\frac 12$} In this case, \eqref{1tr}
gives $-2<\tr \frakX_0(0) < 2$, and 
Theorem~\ref{thm:perI} ensures 
$\sigmaAC(J^\pm )=\RR$. 
\subsubsection{Case $g=\frac 12$} In this case 
 $\tr \frakX_0(0)=2$,   $\frakX_0(0)$ is not 
diagonalizable and, using   
 Proposition~\ref{prop:1} with $s+t=2\kappa +1$, 
 we obtain \eqref{eq:2} with 
\[
    s_n =\alpha_{n-1} \equiv g=\frac{1}{2}
    \] 
and    \[     r_n^\pm =\frac{\beta_n}{2}(s+t) -\gamma_n^\pm = 
    \kappa + \frac{1}{2} \mp (-1)^n \frac{\Delta}{2}.
\]
Since  
$1-\varepsilon [\frakX_0(0)]_{1,1}=2$ and 
$[\frakX_0(0)]_{2,1}=\frac 1g =2$, we find that  \eqref{eq:3} gives
\[
    \tau^\pm (x) =4-4(2x+r_0^\pm +r_1^\pm )=
    -8(x+ \kappa ).
\]
Therefore $\tau^\pm (x) < 0 \Leftrightarrow x> -\kappa$ and the conclusion follows from Theorem~\ref{thm:perIIb}.

\subsection{Proof of Theorem \ref{thm:22}} \label{sec:42} 
 In this section, $H$ is the Hamiltonian of the  
two-photon Rabi model defined as 
the closure of $\hat H$ given in \eqref{22}. 
\par  
If $e^\nu_n :=e^\nu \otimes e_n$ as before     
and $\mu \in \{0,1\}$, then 
$$
g\big(\sigma_x \otimes (\hat a^\dagger )^2 \big) 
e_{2n+\mu}^{\nu} =a_{\mu ,n}\, 
e^{-\nu}_{2(n+1)+\mu}, 
$$ 
$$
g\big(\sigma_x \otimes \hat a^2 \big) 
e_{2n+\mu}^{\nu} =a_{\mu ,n-1} \, 
e^{-\nu}_{2(n-1)+\mu}
$$ 
with $(a_{\mu ,n} )$ given by \eqref{thm2}. Moreover, 
$$
 \Big( I_2 \otimes \hat N +\frac{\Delta}{2}\, \sigma_z \otimes I\Big) e_{2n+\mu}^{\nu} =
\left( {2n+\mu +\nu \frac{\Delta}{2} 
 }\right) e_{2n+\mu}^{\nu} 
 $$
and, introducing 
$\mathfrak{F}_\mu^\pm :=\{ 
f_{\mu ,n}^\pm \}_{n\in \NN_0}$
with 
\begin{equation} \label{42f} 
 f_{\mu ,n}^\pm :=e_{2n+\mu}^{\pm (-1)^n}, 
\end{equation} 
we obtain
\begin{equation} \label{42H}
\hat H f_{\mu ,n}^\pm =a_{\mu ,n} f_{\mu ,n+1}^\pm +a_{\mu ,n-1}f^\pm_{\mu ,n-1} 
+ b^{\pm}_{\mu ,n} f^\pm_{\mu ,n} ,  
\end{equation} 
where $(b^\pm_{\mu ,n} )$ is given by \eqref{thm2} and $a_{\mu ,j}f^\pm_{\mu ,j}:=0$ when $j<0$.
In order to obtain \eqref{decomp2}, we consider 
the orthogonal decomposition 
\begin{equation} \label{42orth}
\CC^2\otimes \ell^2(\NN_0) =\mathcal H_0^- 
\oplus \mathcal H_0^+ \oplus \mathcal H_1^- 
\oplus \mathcal H_1^+
\end{equation} 
with $\mathcal H_\mu^\pm$ defined as the closure of 
${\rm span}(\mathfrak F_\mu^\pm )$. 
Due to \eqref{42H}, $\mathcal H_\mu^\pm$ are 
invariant subspaces for $\hat H$ and 
 $J((a_{\mu ,n}),(b_{\mu ,n}^\pm ))$ 
is the matrix of $\hat H|_{\mathcal H_\mu^\pm}$
in the basis $\mathfrak F_\mu^\pm$. 
  \par 
As before, Proposition~\ref{prop:1} and \ref{prop:2} 
ensure that \eqref{eq:1},  \eqref{eq:1'} are satisfied
with $N=2$ and 
\begin{equation} \label{42alpha}
   a_{\mu ,n} = \alpha_n 
    \sqrt{(n+\tfrac{1}{2} + \tfrac{\mu}{2}) (n+1+\tfrac{\mu}{2})} \hbox{ \, with \, } 
    \alpha_n \equiv 2g,  
\end{equation}
\begin{equation} \label{beta}
 b_{\mu ,n}^\pm =\beta_n n +\gamma_{\mu ,n}^\pm  
\hbox{ \, with \, } 
\beta_n \equiv 2, \hspace{3mm} \gamma_{\mu ,n}^\pm = \mu \pm (-1)^n \frac{\Delta}{2} . 
\end{equation}
 Moreover, ${\mathfrak B}_n(0)$,  $\frakX_n(0)$ and  
${\rm tr} \frakX_n(0)$ are given by \eqref{eq:100} 
 and \eqref{1tr}, respectively. 
This gives the following three cases similar to before. 
\subsubsection{Case $0<g<\frac 12$} We have 
$\tr \frakX_0(0) > 2$ and 
Theorem~\ref{thm:perIII} ensures 
$\sigma(J_\mu^\pm )=\sigmaD(J_\mu^\pm )$. 
\subsubsection{Case $g>\frac 12$} We have 
$-2<\tr \frakX_0(0) < 2$ and 
Theorem~\ref{thm:perI} ensures 
$\sigmaAC(J_\mu^\pm )=\RR$. 
\subsubsection{Case $g=\frac 12$} 
We have $\tr \frakX_0(0)=2$ and, using   
 Proposition~\ref{prop:1} with $s+t=\mu +\frac 32$, 
 we obtain \eqref{eq:2} with
\[
    s_n =\alpha_{n-1} \equiv 2g=1 
    \] 
    and
    \[ 
r_n^\pm = \frac{\beta_n}{2}\Big( \, \mu +\frac 32 \, \Big) 
-\gamma_{\mu ,n}^\pm =\frac{3}{2} \mp (-1)^n \frac{\Delta}{2} . 
 \]
Therefore,  
 \[
\tau^\pm (x)=4-2(2x+r_0^\pm +r_1^\pm )=-4x-2. 
\] 
It is clear that $\tau^\pm (x) < 0 \Leftrightarrow x> 
-\frac 12$ and the conclusion follows from Theorem~\ref{thm:perIIb}. 
\subsection{Proof of Theorem \ref{thm:23}} 
\label{sec:43}
In this section, $H$ is the Hamiltonian of the anisotropic two-photon Rabi model defined as 
the closure of $\hat H$ given in \eqref{23}. 
\par 

Using $g$ and $g'$ given by \eqref{thm3g}, we find
$$
(g_- \sigma_- +g_+ \sigma_+ )e^{\pm 1} = g_\mp e^{\mp 1} 
=(g \mp g') e^{\mp 1} 
$$  
and, using $f_{\mu ,n}^\pm =e^{\pm (-1)^n}_{2n+\mu}$, we 
can express
$$ 
\big( (g_-\sigma_- +g_+\sigma_+) \otimes (\hat a^\dagger )^2 
\big) f_{\mu ,n}^\pm = a_{\mu ,n}^\pm f_{\mu ,n+1}^\pm , 
$$ 
$$ 
\big( (g_+\sigma_- +g_-\sigma_+) \otimes \hat a^2 
\big) f_{\mu ,n}^\pm = a_{\mu ,n-1}^\pm f_{\mu ,n-1}^\pm  
$$
with $a_{\mu ,n}^\pm$  given by 
\eqref{thm3}. Moreover, 
\begin{equation} \label{43H}
\hat H f_{\mu ,n}^\pm =a_{\mu ,n}^\pm f_{\mu ,n+1}^\pm +a_{\mu ,n-1}^\pm f^\pm_{\mu ,n-1} 
+ b^{\pm}_{\mu ,n} f^\pm_{\mu ,n}  
\end{equation} 
with  $b_{\mu ,n}^\pm$ given by 
\eqref{thm3} 
and  $a^\pm_{\mu ,j}f^\pm_{\mu ,j}:=0$ when $j<0$. 
It is clear that we obtain \eqref{decomp2}, using 
\eqref{42orth} similarly as in Section 
\ref{sec:22}. Then Proposition~\ref{prop:1} and \ref{prop:2} 
ensure that \eqref{eq:1},  \eqref{eq:1'} are satisfied
with $N=2$,  
$$
a_{\mu ,n}^\pm = \alpha_n^\pm  
    \sqrt{(n+\tfrac{1}{2} + \tfrac{\mu}{2}) (n+1+\tfrac{\mu}{2})} \hbox{ \, with \, } 
    \alpha_n^\pm =  2(g \mp (-1)^n g'), 
$$ 
and $\beta_n$, $\gamma_{\mu ,n}^\pm$ as in \eqref{beta}.  
Moreover, 
 $$
\frakB_n^\pm (0) =
    \begin{pmatrix}
        0 & 1 \\
    -\frac{\alpha^\pm_{n+1}}{\alpha^\pm_n} 
    & -\frac{2}{\alpha^\pm_n}
    \end{pmatrix}, \quad 
   \frakX_n^\pm (0) =
    \begin{pmatrix}
        \frac{-\alpha^\pm_{n+1}}{\alpha^\pm_n} & 
        -\frac{2}{\alpha^\pm_n} \\
     \frac{2}{\alpha^\pm_n} & 
\frac{4-(\alpha^\pm_n)^2}{\alpha^\pm_n \alpha^\pm_{n+1}}
    \end{pmatrix} 
 $$ 
If $\sigma \in \{ -1, 1\}$, then 
\[
    \tr \frakX^\pm_0(0) -2\sigma =
\frac{4-(\alpha^\pm_0)^2-(\alpha^\pm_1)^2}{\alpha^\pm_0 \alpha^\pm_1} - 2\sigma  =
\frac{4-(\alpha^\pm_0 +\sigma \alpha^\pm_1)^2}{\alpha^\pm_0 \alpha^\pm_1}.
\] 
\subsubsection{Case $g<\frac 12$} In this case 
$4-(\alpha^\pm_0 +\alpha^\pm_1)^2=4-16g^2>0$, therefore 
${\rm tr\,} \frakX^\pm_0(0) -2>0$, i.e.  
Theorem~\ref{thm:perIII} ensures 
$\sigma(J_\mu^\pm )=\sigmaD(J_\mu^\pm )$. 
\subsubsection{Case $|g'|>\frac 12$} 
In this case 
$4-(\alpha^\pm_0 -\alpha^\pm_1)^2=4-16g'^2<0$, therefore  
${\rm tr\,} \frakX^\pm_0(0) +2<0$, i.e. 
Theorem~\ref{thm:perIII} ensures 
$\sigma(J_\mu^\pm )=\sigmaD(J_\mu^\pm )$. 
\subsubsection{Case $|g'|<\frac 12 <g$}  In this case 
$-2<{\rm tr\,} \frakX^\pm_0(0)<2$, i.e.  
Theorem~\ref{thm:perI} ensures 
$\sigmaAC(J_\mu^\pm )=\RR$. 
\subsubsection{Case $g=\frac 12$} In this case 
${\rm tr\,} \frakX^\pm_0(0)=2$, $\frakX^\pm_0(0)$ 
is not diagonalizable and, using   
 Proposition~\ref{prop:1} with $s+t=\mu +\frac 32$, 
 we obtain \eqref{eq:2} with 
\[
s_n^\pm =\alpha^\pm_{n-1} = 2 
\big( g \pm (-1)^n g' \big) , \quad \beta_n \equiv 2 , \hspace{4mm} \gamma_{\mu ,n}^\pm = \mu \pm (-1)^n \frac{\Delta}{2}
\] 
and 
\[ 
     \frac{\beta_n}{2}\Big( \, \mu +\frac 32 \, \Big) 
    -\gamma_{\mu ,n}^\pm =
 \frac{3}{2} \mp (-1)^n \frac{\Delta}{2} =r_n^\pm . 
\]
Since $\varepsilon =1$ and 
$\alpha^\pm_0+\alpha^\pm_1 =4g=2$, we find 
$$
2-\varepsilon [\frakX^\pm_0(0)+\frakX^\pm_1(0)]_{1,1} 
=2+\frac{\alpha^\pm_1}{\alpha^\pm_0}+
\frac{\alpha^\pm_0}{\alpha^\pm_1} =
\frac{(\alpha^\pm_0+\alpha^\pm_1)^2}{\alpha^\pm_0\alpha^\pm_1} =
\frac{4}{\alpha^\pm_0\alpha^\pm_1}  .
$$
Moreover, $$
\frac{[\frakX^\pm_n(0)]_{2,1}}{\alpha^\pm_{n-1}} = \frac{2}{\alpha^\pm_0\alpha^\pm_1}
$$ 
and 
\[
    \tau^\pm (x) = \frac{4}{\alpha^\pm_0 \alpha^\pm_1} 
   -\frac{2}{\alpha^\pm_0 \alpha^\pm_1} 
   (2x+r_0^\pm +r_1^\pm ) 
   =-\frac{4x+2}{\alpha^\pm_0 \alpha^\pm_1} .
\]
Therefore $\tau^\pm (x) < 0 \Leftrightarrow x> 
-\frac 12$ and the conclusion follows from Theorem~\ref{thm:perIIb}. 
\subsubsection{Case $|g'|=\frac 12$} In this case 
${\rm tr\,} \frakX^\pm_0(0)=-2$, $\frakX^\pm_0(0)$ 
is not diagonalizable, $\varepsilon =-1$ and 
$$
2-\varepsilon [\frakX^\pm_0(0)+\frakX^\pm_1(0)]_{1,1} 
=2-\frac{\alpha^\pm_1}{\alpha^\pm_0}-
\frac{\alpha^\pm_0}{\alpha^\pm_1} =
-\frac{(\alpha^\pm_0-\alpha^\pm_1)^2}{\alpha^\pm_0\alpha^\pm_1} =
-\frac{4}{\alpha^\pm_0\alpha^\pm_1}  , 
$$
hence
\[
\tau^\pm (x) = -\frac{4}{\alpha^\pm_0 \alpha^\pm_1} 
   +\frac{2}{\alpha^\pm_0 \alpha^\pm_1} 
   (2x+r_0^\pm +r_1^\pm ) 
   =\frac{4x+2}{\alpha^\pm_0 \alpha^\pm_1} 
\]
Therefore $\tau^\pm (x) < 0 \Leftrightarrow x< 
-\frac{1}{2}$ and the conclusion follows from Theorem~\ref{thm:perIIb}. 
\subsection{Proof of Theorem \ref{thm:24}} 
\label{sec:44}
 In this section, $H$ is the Hamiltonian of the 
two-photon Rabi--Stark model defined as 
the closure of $\hat H$ given in \eqref{24}. 
\par  
Consider     
\eqref{12orth} with $\mathcal{H}_\mu^\pm$ 
generated by $\mathfrak{F}_\mu^\pm :=\{ 
f_{\mu ,n}^\pm \}_{n\in \NN_0}$ defined in 
\eqref{42f}. Then it is easy to check that \eqref{42H} holds true with 
$(a_{\mu ,n})$,  $(b_{\mu ,n}^\pm )$ given by 
\eqref{thm4}. Thus,   
 $\mathcal H_\mu^\pm$ are 
invariant subspaces for $\hat H$ and 
 $J((a_{\mu ,n}),(b_{\mu ,n}^\pm ))$ 
is the matrix of $\hat H|_{\mathcal H_\mu^\pm}$
in the basis $\mathfrak F_\mu^\pm$. 
Then Proposition~\ref{prop:1} and \ref{prop:2}  \eqref{eq:1},  \eqref{eq:1'} are satisfied
with $N=2$, $\alpha_n$ as in \eqref{42alpha} and 
$$
b_{\mu ,n}^\pm =\beta_n^\pm n +\gamma_{\mu ,n}^\pm  
\hbox{ \, with \, } 
\beta_n^\pm = 2(1\pm (-1)^n \kappa ), \hspace{2mm} 
\gamma_{\mu ,n}^\pm = (1\pm (-1)^n \kappa ) \mu 
    \pm (-1)^n \frac{\Delta}{2}
$$ 
Moreover, we have 
\[
    \frakB_{n}^\pm (0) =
    \begin{pmatrix}
        0 & 1 \\
        -1 & -\frac{1\pm (-1)^n\kappa }{g}
    \end{pmatrix} ,
\quad 
    \frakX_{n}^\pm (0) =
    \begin{pmatrix}
        -1 & \frac{-1\mp (-1)^n\kappa }{g} \\
       \frac{1\mp (-1)^n\kappa }{g} & 
        -1+\frac{1-\kappa^2 }{g^2}
    \end{pmatrix}, 
\]
hence 
\begin{equation} \label{4tr}
 {\rm tr}\frakX_{n}^\pm (0) =-2+
 \frac{1-\kappa^2 }{g^2}.   
\end{equation} 
\subsubsection{Case $|\kappa | >1$} In this case, \eqref{4tr}
gives $\tr \frakX_0^\pm (0) < -2$ and 
Theorem~\ref{thm:perIII} ensures 
$\sigma(J_\mu^\pm )=\sigmaD(J_\mu^\pm )$.
\subsubsection{Case $\kappa^2 + 4g^2<1$} In this case, \eqref{4tr}
gives $\tr \frakX_0^\pm (0)>2$ and 
Theorem~\ref{thm:perI} ensures 
$\sigma(J_\mu^\pm )=\sigmaD(J_\mu^\pm )$.
\subsubsection{Case $|\kappa | < 1$ and $\kappa^2 + 4g^2>1$} 
In this case, \eqref{4tr}
gives $|\tr \frakX_0^\pm (0)|< 2$ and 
Theorem~\ref{thm:perIII} ensures 
$\sigmaAC(J_\mu^\pm )=\RR$. 

\subsubsection{Case $\kappa^2 + 4g^2 =1$ and $0<g<\frac 12$} 
In this case, \eqref{4tr}
gives $\tr \frakX^\pm_0(0)=2$, 
$\frakX^\pm_0$ is not diagonalizable, and using   
 Proposition~\ref{prop:1} with $s+t=\mu +\frac 32$, 
 we obtain \eqref{eq:2} with 
\[
s_n =\alpha_{n-1} \equiv 2g  \] 
and 
\[
 \frac{\beta^\pm_n}{2}\Big( \, \mu +\frac 32 \, \Big) 
    -\gamma_{\mu ,n}^\pm = \frac{3}{2} (1 \pm (-1)^n \kappa) \mp (-1)^n \frac{\Delta}{2} =r_n^\pm .
\] 
 Then $1-\varepsilon [\frakX^\pm_n(0)]_{1,1}=2$ and 
we obtain 
\[ 
\tau^\pm (x) =-\frac{x}{g^2} + 4 - \frac{3(1-\kappa^2)}{2g^2} - \frac{\kappa \Delta}{2 g^2}
\]
and $\tau^\pm (x) < 0 \Leftrightarrow x> 
\frac {\kappa^2 -1- \kappa \Delta}{2}$ and the conclusion follows from Theorem~\ref{thm:perIIb}. 
\subsubsection{Case $\kappa =\pm 1$} In this case, \eqref{4tr} gives $\tr \frakX^\pm_0(0)=-2$ and  
$1-\varepsilon [\frakX^\pm_n(0)]_{1,1}=0$, hence
\[ 
    \tau^\pm (x) = \frac{x}{g^2} + \frac{\kappa \Delta}{2g^2}
\]
Therefore $\tau^\pm (x) < 0 \Leftrightarrow x < 
- \frac {\kappa \Delta}{2}$ and the conclusion follows from Theorem~\ref{thm:perIIb}. 
\section{Conclusions} \label{sec:6} 
In this paper, we have investigated spectral transitions for Rabi models which are unitarily similar to direct sums of Jacobi operators. We have presented an analysis of the spectrum of 
Jacobi operators based on the subordinacy theory developed 
in \cite{PeriodicII, SwiderskiTrojan2019, Discrete, jordan, jordan2}. In particular, new results have been given for the intensity-dependent Rabi model. We have proved 
that the spectral transition in the intensity-dependent Rabi model is similar to that in the two-photon model; the only difference is that the critical coupling gives the essential spectrum, which is a half-line depending on the additional parameter $\kappa$ of the model.\, 
Concerning the two-photon anisotropic Rabi model, we have  proved that there are two cases of critical coupling: 
the well-known case of the essential spectrum  $\sigmaEss (H)=[-\frac 12, \infty )$ and the case of the essential spectrum 
$\sigmaEss (H)=(-\infty ,-\frac 12 ]$, which 
seems to be a new result. 
Concerning the two-photon Rabi-Stark model, 
we have proved that there are two cases of critical coupling: the case $4g^2+\kappa^2=1$ with 
$\kappa \in (-1,1)$ and the case $\kappa =\pm 1$. 
The essential spectrum is formed 
by a half-line bounded below 
in the first case and by a half-line bounded above in the second case.  
We note that $\inf \sigmaEss (H)$ in the first case and $\sup \sigmaEss (H)$ in the second case depend on the parameters of the model. Our result in 
the case $\kappa =\pm 1$ seems to be new for 
the two-photon Rabi-Stark, but an analogical case  
 was considered for the one-photon Rabi-Stark. 
 \par 
We have also proved that 
there is no singular spectrum. Moreover, we 
have proved that there is no 
eigenvalue in the interior of the continuous spectrum for all models considered. This fact 
is particularly 
interesting for the two-photon Rabi-Stark model 
in the context of the work \cite{Braak-Stark}, 
investigating the question of existence of "bound states embedded in the continuous spectrum" for the one-photon Rabi-Stark model. \par 
 The subordinacy theory is a powerful 
tool to locate the continuous spectrum 
and to ensure absence of eigenvalues in the interior of the continuous 
spectrum. On the other hand, it gives 
little information on the discrete spectrum. 
For this reason, our descriptions of spectral 
transitions are essentially limited to the statements about absence or presence of 
the continuous spectrum. In particular, 
this approach cannot treat superradiant problems 
or phase transitions connected with the behaviour 
of the ground state. 
We note that numerical analysis based on a truncated 
Hilbert space does not work well if the model 
approaches the spectral transition.  An analysis of 
the associated $G$-functions 
is the only method that can detect 
the first-order phase transition 
in the anisotropic two-photon Rabi 
model (see \cite{Xie-Stark20}) and 
the second-order phase transition 
in the 
two-photon Rabi-Stark model 
(see \cite{Li_25}). However, 
there are perspectives to 
develop a subordinacy theory 
of block Jacobi matrices 
that could be applied to the mixed Rabi model in order to 
describe phase transitions 
considered in quantum metrology (see \cite{Ying_22}).
\vspace{2mm} 
\begin{bibliography}{rabi}
\bibliographystyle{amsplain}

\providecommand{\bysame}{\leavevmode\hbox to3em{\hrulefill}\thinspace}
\providecommand{\MR}{\relax\ifhmode\unskip\space\fi MR }
\providecommand{\MRhref}[2]{%
  \href{http://www.ams.org/mathscinet-getitem?mr=#1}{#2}
}
\providecommand{\href}[2]{#2}
\begin{thebibliography}{10}

\bibitem{Armenta}
R.~J. Armenta~Rico, F.~H. Maldonado-Villamizar, and B.~M. Rodriguez-Lara,
  \emph{Spectral collapse in the two-photon quantum {R}abi model}, Phys. Rev. A
  \textbf{101} (2020), 063825.

\bibitem{Ber}
P.~Bertet, S.~Osnaghi, P.~Milman, A.~Auffeves, P.~Maioli, M.~Brune, J.~M.
  Raimond, and S.~Haroche, \emph{Generating and {P}robing a {T}wo-{P}hoton
  {F}ock {S}tate with a {S}ingle {A}tom in a {C}avity}, Phys. Rev. Lett.
  \textbf{88} (2002), 143601.

\bibitem{Braak:det}
D.~Braak, \emph{Spectral {D}eterminant of the {T}wo-{P}hoton {Q}uantum {R}abi
  {M}odel}, Annalen der Physik \textbf{536} (2024), no.~6, 2200519.

\bibitem{Braak-Stark}
D.~Braak, L.~Cong, H.-P. Eckle, H.~Johannesson, and E.~K. Twyeffort,
  \emph{Spectral continuum in the {R}abi-{S}tark model}, J. Opt. Soc. Am. B
  \textbf{41} (2024), no.~8, C97--C111.

\bibitem{Bru}
M.~Brune, F.~M. Schmidt-Kaler, A.~Maali, J.~Dreyer, E.~Hagley, J.~M. Raimond,
  and S.~Haroche, \emph{Quantum {R}abi {O}scillation: {A} {D}irect {T}est of
  {F}ield {Q}uantization in a {C}avity}, Phys. Rev. Letters \textbf{76} (1996),
  no.~11, 1800--1803.

\bibitem{Chen-Stark}
X.-Y. Chen, Y-F. Xie, and Q.-H. Chen, \emph{Quantum criticality of the
  {R}abi-{S}tark model at finite frequency ratios}, Phys. Rev. A \textbf{102}
  (2020), 063721.

\bibitem{Polaron}
L.~Cong, X.-M. Sun, M.~Liu, Z.-J. Ying, and H.-G. Luo, \emph{Polaron picture of
  the two-photon quantum {R}abi model}, Phys. Rev. A \textbf{99} (2019),
  013815.

\bibitem{Cui}
S.~Cui, J.-P. Cao, H.~Fan, and L.~Amico, \emph{Exact analysis of the spectral
  properties of the anisotropic two-bosons {R}abi model}, Journal of Physics A:
  Mathematical and Theoretical \textbf{50} (2017), no.~20, 204001.

\bibitem{Va}
E.~Del~Valle, S.~Zippilli, F.~P. Laussy, A.~Gonzalez-Tudela, G.~Morigi, and
  C.~Tejedor, \emph{Two-photon lasing by a single quantum dot in a high-{Q}
  microcavity}, Phys. Rev.B \textbf{81} (2010), 035302.

\bibitem{Duan_2022}
L.~Duan, \emph{Unified approach to the nonlinear {R}abi models}, New Journal of
  Physics \textbf{24} (2022), no.~8, 083045.

\bibitem{Duan2015}
L.~Duan, S.~He, D.~Braak, and Q.-H. Chen, \emph{Solution of the two-mode
  quantum {R}abi model using extended squeezed states}, {EPL} (Europhysics
  Letters) \textbf{112} (2015), no.~3, 34003.

\bibitem{Duan2016}
L.~Duan, Y.-F. Xie, D.~Braak, and Q.-H. Chen, \emph{Two-photon {R}abi model:
  analytic solutions and spectral collapse}, J. Phys. A \textbf{49} (2016),
  no.~46, 464002, 13.

\bibitem{mixed}
L.~Duan, Y.-F. Xie, and Q.-H. Chen, \emph{The mixed quantum {R}abi model},
  Scientific Reports \textbf{9} (2019), no.~1, 18353.

\bibitem{Johan}
H.-P. Eckle and H.~Johannesson, \emph{A generalization of the quantum {R}abi
  model: exact solution and spectral structure}, Journal of Physics A:
  Mathematical and Theoretical \textbf{50} (2017), no.~29, 294004.

\bibitem{Felicetti_2019}
S.~Felicetti and A.~Boité, \emph{Observing the spectral collapse of two-photon
  interaction models}, 11th Italian Quantum Information Science conference
  (IQIS2018), IQIS 2018, MDPI, 2019, p.~41.

\bibitem{Feli1}
S.~Felicetti, M.~J. Hwang, and A.~L. Boité, \emph{Ultrastrong coupling regime
  of non-dipolar light-matter interactions}, Phys. Rev. A \textbf{98} (2018),
  053859.

\bibitem{Feli}
S.~Felicetti, J.~S. Pedernales, I.~L. Egusquiza, G.~Romero, L.~Lamata,
  D.~Braak, and E.~Solano, \emph{Spectral collapse via two-phonon interactions
  in trapped ions}, Phys. Rev. A \textbf{92} (2015), 033817.

\bibitem{Feli2}
S.~Felicetti, D.~Z. Rossatto, E.~Rico, E.~Solano, and Forn-Díaz P.,
  \emph{Two-photon quantum {R}abi model with superconducting circuits}, Phys.
  Rev. A \textbf{97} (2018), 013851.

\bibitem{G}
C.~Gerry, \emph{Two-photon {J}aynes--{C}ummings model interacting with the
  squeezed vacuum}, Phys. Rev. A (3) \textbf{37} (1988), no.~7, 2683--2686.

\bibitem{Hammani}
M.~Hammani, Z.~Sakhi, and M.~Bennai, \emph{On the two-photon quantum {R}abi
  model at the critical coupling strength}, Optical and Quantum Electronics
  \textbf{56} (2023), no.~1, 102.

\bibitem{JANAS}
J.~Janas and S.~Naboko, \emph{Spectral analysis of selfadjoint jacobi matrices
  with periodically modulated entries}, Journal of Functional Analysis
  \textbf{191} (2002), no.~2, 318--342.

\bibitem{Li-ani}
J.~Li, D.~Braak, and Q.-H. Chen, \emph{Critical spectrum of the anisotropic
  two-photon quantum {R}abi model}, Phys. Rev. A \textbf{111} (2025), 043706.

\bibitem{Li-Stark}
J.~Li and Q.-H. Chen, \emph{Two-photon {R}abi-{S}tark model}, Journal of
  Physics A: Mathematical and Theoretical \textbf{53} (2020), 315301.

\bibitem{Li_25}
J.~Li and Q.‐H. Chen, \emph{Critical spectrum and quantum criticality in the
  two‐photon {R}abi–{S}tark model}, Advanced Quantum Technologies
  \textbf{9} (2025), no.~3.

\bibitem{Li_26}
J.~Li, J.-L. Wang, Chen Q.-H., and Lin H.-Q., \emph{Quantum criticality from
  spectral collapse in the two-photon {R}abi model}, 2026.

\bibitem{Lo!}
C.~F. Lo, \emph{Demystifying the spectral collapse in two-photon {R}abi model},
  Scientific Reports \textbf{10} (2020), no.~1, 14792.

\bibitem{Lo2020'}
\bysame, \emph{Manipulating the spectral collapse in two-photon {R}abi model},
  Scientific Reports \textbf{10} (2020), no.~1, 18761.

\bibitem{Lo2021'}
\bysame, \emph{Deciphering the spectral collapse in two-photon {R}abi model},
  Scientific Reports \textbf{11} (2021), no.~1, 10647.

\bibitem{Lo-ani}
\bysame, \emph{Spectral collapse in anisotropic two-photon {R}abi model},
  Scientific Reports \textbf{11} (2021), 12401.

\bibitem{Lo2021'''}
\bysame, \emph{Spectral collapse in multiqubit two-photon {R}abi model},
  Scientific Reports \textbf{11} (2021), no.~1, 5409.

\bibitem{Lo2021''}
\bysame, \emph{Spectral collapse in two-mode two-photon {R}abi model}, Physica
  A: Statistical Mechanics and its Applications \textbf{573} (2021), 125921.

\bibitem{Lo2022}
\bysame, \emph{Spectral collapse in mixed {R}abi model}, Physica A: Statistical
  Mechanics and its Applications \textbf{603} (2022), 127678.

\bibitem{Lupo}
E.~Lupo, A.~Napoli, A.~Messina, E.~Solano, and Í.~L. Egusquiza, \emph{A
  continued fraction based approach for the {T}wo-photon {Q}uantum {R}abi
  {M}odel}, Scientific Reports \textbf{9} (2019), no.~1, 4156.

\bibitem{Maciej-Stark}
A.~J. Maciejewski, M.~Przybylska, and T.~Stachowiak, \emph{An exactly solvable
  system from quantum optics}, Physics Letters A \textbf{379} (2015),
  no.~24-25, 1503–1509.

\bibitem{Maciej}
A.~J. Maciejewski and T.~Stachowiak, \emph{A novel approach to the spectral
  problem in the two photon rabi model}, Journal of Physics A: Mathematical and
  Theoretical, \textbf{50} (2017), 244003.

\bibitem{NG}
K.~M. Ng, C.~F. Lo, and K.L. Liu, \emph{Exact eigenstates of the two-photon
  {J}aynes-{C}ummings model with the counter-rotating term}, The European
  Physical Journal D - Atomic, Molecular, Optical and Plasma Physics \textbf{6}
  (1999), no.~1, 119--126.

\bibitem{NG2000463}
K.M. Ng, C.~F. Lo, and K.~L. Liu, \emph{Exact eigenstates of the
  intensity-dependent {J}aynes–{C}ummings model with the counter-rotating
  term}, Physica A \textbf{275} (2000), no.~3, 463--474.

\bibitem{Ota}
Y.~Ota, S.~Iwamoto, N.~Kumagai, and Y.~Arakawa, \emph{Spontaneous
  {T}wo-{P}hoton {E}mission from a {S}ingle {Q}uantum {D}ot}, Physical Review
  Letters \textbf{107} (2011), no.~23, 233602.

\bibitem{Penna}
V.~Penna, F.~A. Raffa, and R.~Franzosi, \emph{Algebraic properties and spectral
  collapse in nonlinear quantum {R}abi models}, J. Phys. A: Math. Theor.
  \textbf{51} (2018), 045301.

\bibitem{Puebla}
R.~Puebla, M.-J. Hwang, J.~Casanova, and M.~B. Plenio, \emph{Protected
  ultrastrong coupling regime of the two-photon quantum {R}abi model with
  trapped ions}, Physical Review A \textbf{95} (2017), no.~6, 063844.

\bibitem{rabi1936process}
I.~I. Rabi, \emph{On the process of space quantization}, Physical Review
  \textbf{49} (1936), no.~4, 324.

\bibitem{rabi1937space}
\bysame, \emph{Space quantization in a gyrating magnetic field}, Physical
  Review \textbf{51} (1937), no.~8, 652.

\bibitem{Rodriguez-Lara:14}
B.~M. Rodr\'{i}guez-Lara, \emph{Intensity-dependent quantum {R}abi model:
  spectrum, supersymmetric partner, and optical simulation}, J. Opt. Soc. Am. B
  \textbf{31} (2014), no.~7, 1719--1722.

\bibitem{Schm}
K.~Schm\"{u}dgen, \emph{Unbounded self-adjoint operators on hilbert space},
  Graduate Texts in Mathematics, vol. 265, Springer Netherlands, 2012.

\bibitem{Schmudgen2017}
\bysame, \emph{The moment problem}, Graduate Texts in Mathematics, vol. 277,
  Springer, Cham, 2017.

\bibitem{Stufler}
S.~Stufler, P.~Machnikowski, P.~Ester, M.~Bichler, V.~M. Axt, T.~Kuhn, and
  A.~Zrenner, \emph{Two-photon {R}abi oscillations in a single
  $\mathrm{{I}n}_{x}
  \mathrm{{G}a}_{1\ensuremath{-}x}\mathrm{As}/\mathrm{Ga}\mathrm{As}$ quantum
  dot}, Phys. Rev. B \textbf{73} (2006), 125304.

\bibitem{PeriodicII}
G.~\'{S}widerski, \emph{Periodic perturbations of unbounded {J}acobi matrices
  {II}: {F}ormulas for density}, J. Approx. Theory \textbf{216} (2017), 67--85.

\bibitem{SwiderskiTrojan2019}
G.~\'{S}widerski and B.~Trojan, \emph{Asymptotics of orthogonal polynomials
  with slowly oscillating recurrence coefficients}, J. Funct. Anal.
  \textbf{278} (2020), no.~3, 108326, 55.

\bibitem{Discrete}
\bysame, \emph{About essential spectra of unbounded {J}acobi matrices}, J.
  Approx. Theory \textbf{278} (2022), Paper No. 105746, 47.

\bibitem{jordan2}
\bysame, \emph{Orthogonal polynomials with periodically modulated recurrence
  coefficients in the {J}ordan block case {II}}, Constr. Approx. \textbf{58}
  (2023), no.~3, 615--686.

\bibitem{jordan}
\bysame, \emph{Orthogonal polynomials with periodically modulated recurrence
  coefficients in the {J}ordan block case}, Ann. Inst. Fourier (Grenoble)
  \textbf{74} (2024), no.~4, 1521--1601.

\bibitem{Teschl}
G.~Teschl, \emph{Jacobi operators and completely integrable nonlinear
  lattices}, Math. Surv. Monogr., vol.~72, Providence, RI: American
  Mathematical Society, 2000 (English).

\bibitem{xie2017quantum}
Q.~Xie, H.~Zhong, M.~T. Batchelor, and C.~Lee, \emph{The quantum {R}abi model:
  solution and dynamics}, Journal of Physics A: Mathematical and Theoretical
  \textbf{50} (2017), no.~11, 113001.

\bibitem{Xie-Chen}
Y.-F. Xie and Q.-H. Chen, \emph{Exact solutions to the quantum {R}abi-{S}tark
  model within tunable coherent states}, Communications in Theoretical Physics
  \textbf{71} (2019), no.~5, 623.

\bibitem{Xie-Stark20}
Y.-F. Xie, X.-Y. Chen, X.-F. Dong, and Q.-H. Chen, \emph{First-order and
  continuous quantum phase transitions in the anisotropic quantum
  {R}abi-{S}tark model}, Phys. Rev. A \textbf{101} (2020), 053803.

\bibitem{XieStark}
Y.-F. Xie, L.~Duan, and Q.-H. Chen, \emph{Quantum {R}abi–{S}tark model:
  solutions and exotic energy spectra}, Journal of Physics A: Mathematical and
  Theoretical \textbf{52} (2019), no.~24, 245304.

\bibitem{Yan}
Z.~Yan, J.~Cheng, F.~Qiu, R.~Liu, W.~Zhao, and J.~Ma, \emph{Analytical solution
  and spectral structure of the two-photon anisotropic {R}abi-{S}tark model},
  Phys. Scr. \textbf{99} (2024), 075105.

\bibitem{Ying}
Z.-J. Ying, L.~Cong, and X.-M. Sun, \emph{Quantum phase transition and
  spontaneous symmetry breaking in a nonlinear quantum {R}abi model}, Journal
  of Physics A: Mathematical and Theoretical \textbf{53} (2020), no.~34,
  345301.

\bibitem{Ying_22}
Z.-J. Ying, S.~Felicetti, G.~Liu, and D.~Braak, \emph{Critical quantum
  metrology in the non-linear quantum {R}abi model}, Entropy \textbf{24}
  (2022), no.~8, 1015.

\bibitem{Ying_25}
Z.‐J. Ying, H.‐H. Han, B.‐J. Li, S.~Felicetti, and D.~Braak,
  \emph{Critical quantum metrology in a stabilized two‐photon {R}abi model},
  Advanced Quantum Technologies \textbf{8} (2025), no.~11.

\end{thebibliography}
\end{bibliography}

\end{document}